\documentclass[conference, 10pt]{IEEEtran}
\IEEEoverridecommandlockouts
\usepackage{supertech-llncs-light}
\usepackage{booktabs} % for formal tables

\def\BibTeX{{\rm B\kern-.05em{\sc i\kern-.025em b}\kern-.08em
    T\kern-.1667em\lower.7ex\hbox{E}\kern-.125emX}}

\begin{document}

%% we may remove the following format instructions for conformation
%\special{papersize=8.5in,11in}
% \setlength{\pdfpageheight}{\paperheight}
% \setlength{\pdfpagewidth}{\paperwidth}
% \setlength{\belowdisplayskip}{0ex} \setlength{\belowdisplayshortskip}{0ex}
% \setlength{\abovedisplayskip}{0ex} \setlength{\abovedisplayshortskip}{0ex}
% %% set the space before/after in-text figures
% \setlength{\floatsep}{0ex}
% \setlength{\intextsep}{1 \baselineskip} 
% \setlength{\textfloatsep}{1 \baselineskip}
% \setlength{\abovecaptionskip}{0pt} \setlength{\belowcaptionskip}{0ex}

\title{Improving the Space-Time Efficiency of Processor-Oblivious Matrix Multiplication Algorithms
\thanks{Shanghai Natural Science Funding (No. 18ZR1403100)}
}
% \subtitle{Extended Abstract}

\punt{% begin punt
    % Following Author's info is aligned with ACM
\author{Yuan Tang}
\authornote{Corresponding Author. 
Also affiliated with Shanghai Key Lab. of Intelligent Information Processing}
\affiliation{%
    \institution{School of Software, School of Computer Science, Fudan University}
    \city{Shanghai}
    \state{China}
}
\email{yuantang@fudan.edu.cn}
}% end punt
\author{\IEEEauthorblockN{Yuan Tang}
    \IEEEauthorblockA{\textit{School of Computer Science, Fudan University}\\
    Shanghai, P. R. China\\
    yuantang@fudan.edu.cn}
}
\maketitle

\begin{abstract}
% State the problem
% Why it's an interesting problem
% What your solution achieves
% What follows from your solution: Impact of your solution

Classic cache-oblivious parallel matrix multiplication algorithms
achieve optimality either in time or space, but not both, which
promotes lots of research on the best possible balance or
tradeoff of such algorithms.
We study modern processor-oblivious runtime systems and figure 
out several ways to improve algorithm's time bound while still
bounding space and cache requirements to be asymptotically
optimal.
By our study, we give out sublinear time, optimal work, space
and cache algorithms for both general matrix multiplication on a 
semiring and Strassen-like fast algorithm.
Our experiments also show such algorithms have empirical 
advantages over classic counterparts.
Our study provides new insights and research angles on how
to optimize cache-oblivious parallel algorithms
from both theoretical and empirical perspectives.
\end{abstract}

\punt{% begin punt
%% \begin{CCSXML}
%% <ccs2012>
%%  <concept>
%%   <concept_id>10003752.10003809.10011254.10011257</concept_id>
%%   <concept_desc>Theory of computation~Divide and conquer</concept_desc>
%%   <concept_significance>500</concept_significance>
%%  </concept>
%%  <concept>
%%   <concept_id>10003752.10003809.10011254.10011258</concept_id>
%%   <concept_desc>Theory of computation~Dynamic programming</concept_desc>
%%   <concept_significance>500</concept_significance>
%%  </concept>
%%  <concept>
%%   <concept_id>10010147.10010169.10010170.10010171</concept_id>
%%   <concept_desc>Computing methodologies~Shared memory algorithms</concept_desc>
%%   <concept_significance>500</concept_significance>
%%  </concept>
%% </ccs2012>
%% \end{CCSXML}
%% 
%% \ccsdesc[500]{Theory of computation~Divide and conquer}
%% \ccsdesc[500]{Theory of computation~Dynamic programming}
%% \ccsdesc[500]{Computing methodologies~Shared memory algorithms}
}% end punt

% \keywords{
% }

\secput{intro}{Introduction}

%%
%\vspace*{-1em}
\begin{wrapfigure}{R}{0.26\textwidth}
% \vspace{-.82cm}
% \hspace{-.4in}
\caption{Acronyms \& Notations}
\label{fig:symbols}
\scalebox{0.85}{
%\scriptsize{
%\begin{flushright}
\begin{tabular}{|c|p{3.8cm}|}
\toprule
MM & Matrix Multiplication\\
PO & Processor-Oblivious\\
PA & Processor-Aware\\
% DP  & Dynamic Programming\\
% COP & Cache-Oblivious Parallel\\
% COW & Cache-Oblivious Wavefront\\
RWS & Randomized Work-Stealing\\
CAS & Compare-And-Swap\\
% ND  & Nested Dataflow\\
$n$ & Problem dimension\\
$p$ & Processor Count\\
$\epsilon_i$ & small constant\\
$M$ & Cache size\\
$B$ & cache line size\\
$T_1$ & Work\\
$T_\infty$ & Time (span, depth, critical-path length)\\
$T_p$ & Running time on $p$-processor system\\
$T_1/T_{\infty}$ & Parallelism\\
$Q_1$ & Serial cache complexity\\
$Q_p$ & Parallel cache complexity with $p$ threads\\
$\id{a} \parallel \id{b}$ & task \id{b} has \emph{no} dependency on \id{a}\\
$\id{a} \serial \id{b}$ & task \id{b} has \emph{full} dependency on \id{a}\\
% $\id{a} \fire \id{b}$ & task \id{b} has \emph{partial} dependency on \id{a}\\
\bottomrule
\end{tabular}
%\end{flushright}
%}
}
\vspace{-0.4cm}
\end{wrapfigure}
%\vspace{5em}
%

%
\begin{figure*}[!ht]
\centering
\begin{tabular}{ccccc}
\toprule
Algo. & Work ($T_1$) & Time ($T_\infty$) & Space ($S_p$) & Serial Cache ($Q_1$) \\
\midrule
\proc{CO2} & $O(n^3)$ & $O(n)$ & $O(n^2)$ & $O(n^3/(B\sqrt{M}) + n^2/B)$ \\
\proc{CO3} & $O(n^3)$ & $O(\log n)$ & $O(n^3)$ & $O(n^3/B)$ \\
\midrule
\proc{TAR-MM} & $O(n^3)$ & $O(n)$ & $O(n^2 + p b^2)$ & $O(n^3/(B\sqrt{M}) + n^2/B)$ \\
\proc{SAR-MM} & $O(n^3)$ & $O(\log n)$ & $O(p^{1/3} n^2)$ & $O(n^3/(B\sqrt{M}) + n^2/B)$ \\
\proc{STAR-MM} & $O(n^3)$ & $O(\sqrt{p} \log n)$ & $O(n^2)$ & $O(n^3/(B\sqrt{M}) + n^2/B)$ \\
\midrule
\proc{Strassen} & $O(n^{\log_2 7})$ & $O(\log n)$ & $O(n^{\log_2 7})$ & $O(n^{\log_2 7}/B)$ \\
\midrule
\proc{SAR-Strassen} & $O(n^{\log_2 7})$ & $O(\log n)$ & $O(p n^2)$ & $O(n^{\log_2 7}/(B M^{1/2 \log_2 7 - 1}) + n^2/B)$ \\
\proc{STAR-Strassen-1} & $O(p^{0.09} n^{\log_2 7})$ & $O(p^{1/2} \log n)$ & $O(n^2)$ & $O(p^{0.09} n^{\log_2 7}/(B M^{1/2 \log_2 7 - 1}) + p^{1/2} n^2/B)$ \\
\proc{STAR-Strassen-2} & $O(n^{\log_2 7})$ & $O(\log n)$ & $O(p^{1/2 \log_2 7} n^2)$ & $O(n^{\log_2 7}/(B M^{1/2 \log_2 7 - 1}) + p^{1/2 \log_2 7 - 1} n^2/B)$ \\
\bottomrule
\end{tabular}
\caption{Main results of this paper, with comparisons to 
typical prior works. \proc{CO2} stands for the 
MM (Matrix Multiplication) algorithm with $O(n^2)$ space 
(\figref{mm-n2});
\proc{CO3} stands for the MM algorithm with $O(n^3)$ space
(\figref{mm-n3}); $p$ denotes processor count, $b$ is the 
base-case dimension.}
\label{fig:contri-table}
\end{figure*}

It's important to balance space-time requirements for an
algorithm to achieve high performance on modern shared-memory 
multi-core and many-core systems. 
There are two big classes of parallel algorithms. One is
processor-aware (PA), the other is processor-oblivious (PO).

Typical PA algorithms for Matrix Multiplication (MM) 
\cite{AgarwalBaGu95, Cannon69, AggarwalChSn90, DekelNaSa81, Johnsson93, IronyToTi04, ChowdhuryRa08, BallardDeHo11}
seem hard to achieve a perfect load balance on an arbitrary 
number, e.g. a prime number, of processors.
The PA approach usually maps statically MM's computational 
DAG (Directed Acyclic Graph) of dimensions $n$-by-$m$-by-$k$
, which stands for a multiplication of one $n$-by-$k$ matrix 
with one $k$-by-$m$ matrix,  
onto a $2$D processor grid of dimensions
$p^{1/2}$-by-$p^{1/2}$ \cite{Cannon69, ChowdhuryRa08}
, or a $3$D grid of $p^{1/3}$-by-$p^{1/3}$-by-$p^{1/3}$ 
\cite{AgarwalBaGu95}, or even a $2.5$D grid of 
$(p/c)^{1/2}$-by-$(p/c)^{1/2}$-by-$c$
\cite{SolomonikDe11}, where $p$ is processor count and 
$c \in \{1, 2, \ldots, p^{1/3}\}$ is a parameter depending
on the availability of redundant storage. 
There are several concerns on classic PA algorithms.
\begin{enumerate}
    \item There is no guarantee that $p^{1/2}$ or $p^{1/3}$ will 
        be an integer number, which means that some degree of
        under-utilization of processors is unavoidable. 
        In theory, $p$ can even be a 
        prime number. Hence, various tradeoffs need to be 
        explored for a balance.
    \item Obviously we need different shapes of processor grid
        when multiplying different shapes of MM
        for a balance of computation and communication (overall
        cache misses in the case of shared-memory architecture).
        For instances, multiplying
        a $1$-by-$n$ vector with an $n$-by-$1$ vector,
        an $n$-by-$1$ vector with a $1$-by-$n$ vector,
        or a general $n$-by-$m$ matrix with an $m$-by-$k$ matrix
        require different shapes of processor grid for an
        optimal balance in both computation and communication.
        However, the number of factorizations of $p$ is
        usually much less than the number of possible shapes
        of MM, thus prior algorithms of $2$D, $2.5$D, or $3$D 
        are not flexible to adapt. 
\end{enumerate}

The PO approach, on the contrary, just needs to specify the
data and control dependency of computational DAG, then relies 
on a provably efficient
runtime scheduler such as Cilk \cite{IntelCilkPlus10, Leiserson10}
for a dynamic balance. 
Assuming a Randomized Work-Stealing scheduler \cite{BlumofeJoKu95}
, an algorithm just 
needs to bound its critical-path length, also known as (a.k.a.) 
depth, span, or time bound \cite{JaJa92}, to be polylogarithmic, 
i.e. low-depth \cite{BlellochGiSi10}, its parallel running time
is then $O(T_1/p + T_\infty)$ with high probability (w.h.p.) 
\cite{AroraBlPl98}, where
$T_1$ denotes total work and $T_\infty$ denotes critical-path 
length.
Similarly, if an algorithm's serial cache complexity 
is asymptotically optimal, its parallel 
cache complexity can be bounded by \eqref{parallel-cache-bound} 
\cite{AcarBlBl00, SpoonhowerBlGi09}, where $Q_1$ denotes serial
cache bound, $Q_p$ denotes parallel cache bound, $M$ is cache 
size and $B$ is cache line size of the ideal cache model 
\cite{FrigoLePr12}.
\begin{align}
Q_p &= Q_1 + O(p T_\infty M/B) \label{eq:parallel-cache-bound}
\end{align}

From \eqref{parallel-cache-bound}, we can see that for a PO
algorithm to have efficient parallel cache bound, both 
its serial cache bound ($Q_1$) and critical-path
length ($T_\infty$) have to be optimal.
However, if we look at two typical PO MM algorithms in 
\figref{mm-dac}, we can see that one algorithm (\proc{CO2}
in \figref{mm-n2}) has optimal serial cache bound but
sub-optimal (linear) critical-path length and the other 
(\proc{CO3} in \figref{mm-n3}) has optimal critical-path 
length but non-optimal serial cache and space bounds. 
To the best of our knowledge, no existing PO MM algorithm 
achieves both sublinear critical-path length and optimal
serial cache bound.

Let's take a closer look at general MM $C = A \otimes B$ 
on a closed semiring 
$SR = (S, \oplus, \otimes, 0, 1)$, where $S$ is 
a set of elements, $\oplus$ and $\otimes$ are binary operations 
on $S$, and $0$, $1$ are additive and multiplicative identities
, respectively. For simplicity, we discuss only square MM.
General MM can be computed recursively in a divide-and-conquer 
fashion as follows.
At each level of recursion, the computation of an MM of dimension 
$n$ (i.e. multiplication of two $n$-by-$n$ matrices) is divided 
into four equally sized 
quadrants, which require updates from eight sub-MMs of
dimension $n/2$ as shown in \eqref{mm-dac}.
%
% \begin{figure}[!ht]
\begin{align}
    \begin{bmatrix}
        C_{00} & C_{01} \\
        C_{10} & C_{11}
    \end{bmatrix}
    &= 
    \begin{bmatrix}
        A_{00} & A_{01} \\
        A_{10} & A_{11}
    \end{bmatrix} 
    \otimes
    \begin{bmatrix}
        B_{00} & B_{01} \\
        B_{10} & B_{11}
    \end{bmatrix} \nonumber \\
    &=
    \begin{bmatrix}
        A_{00} \otimes B_{00} & A_{00} \otimes B_{01} \\
        A_{10} \otimes B_{00} & A_{10} \otimes B_{01}
    \end{bmatrix} \nonumber \\
    &\qquad \oplus 
    \begin{bmatrix}
        A_{01} \otimes B_{10} & A_{01} \otimes B_{11} \\
        A_{11} \otimes B_{10} & A_{11} \otimes B_{11}
    \end{bmatrix}
\label{eq:mm-dac}
\end{align}
% \end{figure}
%
\noindent
Depending on the availability of extra space, the computation of 
eight sub-MMs can be scheduled to run either completely in 
parallel as shown in \figref{mm-n3} or in two parallel steps 
as \figref{mm-n2} \cite{CormenLeRi09}. More sophisticated 
approaches are studied in the literature and will be discussed in 
\secref{relWork}.

\begin{figure*}
\hspace*{-1cm}
\begin{subfigure}[b]{0.45 \linewidth}
%\scriptsize
\begin{codebox}
% \resizebox{.45\linewidth}{!}{% begin resizebox
\Procname{\proc{CO3}(\id{C}, \id{A}, \id{B})}
\li \Comment $\id{C} \assign \id{A} \times \id{B}$
\li \If $(\func{sizeof}(\id{C}) \leq \const{base\_size})$
\li     \Then \proc{base-kernel}(\id{C}, \id{A}, \id{B})
\li           \Return
        \End
\li $\id{D} \assign \func{alloc}(\func{sizeof}(C))$ \label{li:mm-n3-alloc-D}
\li \Comment Run all $8$ sub-MMs concurrently
\li $\proc{CO3}(\id{C}_{00}, \id{A}_{00}, \id{B}_{00}) \parallel \proc{CO3}(\id{C}_{01}, \id{A}_{00}, \id{B}_{01})$
\li $\parallel \proc{CO3}(\id{C}_{10}, \id{A}_{10}, \id{B}_{00}) \parallel \proc{CO3}(\id{C}_{11}, \id{A}_{10}, \id{B}_{01})$
\li $\parallel \proc{CO3}(\id{D}_{00}, \id{A}_{01}, \id{B}_{10}) \parallel \proc{CO3}(\id{D}_{01}, \id{A}_{01}, \id{B}_{11})$
\li $\parallel \proc{CO3}(\id{D}_{10}, \id{A}_{11}, \id{B}_{10}) \parallel \proc{CO3}(\id{D}_{11}, \id{A}_{11}, \id{B}_{11})$
\li $\serial$ \label{li:co3-sync} \Comment \Sync
\li \Comment Merge matrix \id{D} into \id{C} by addition
\li \func{madd}(\id{C}, \id{D}) \label{li:co3-madd}
\li \func{free} (\id{D})
\li \Return
% }% end resizebox
\end{codebox}
\vspace*{-1 \baselineskip}
\caption{Recursive MM algorithm with $O(n^3)$ space}
\label{fig:mm-n3}
\end{subfigure}
%
% \hfil
%
\hspace*{1cm}
\begin{subfigure}[b]{0.45 \linewidth}
%\scriptsize
\begin{codebox}
% \resizebox{.45\linewidth}{!}{% begin resizebox
\Procname{\proc{CO2}(\id{C}, \id{A}, \id{B})}
\li \Comment $\id{C} \assign \id{A} \times \id{B}$
\li \If $(\func{sizeof}(\id{C}) \leq \const{base\_size})$
\li     \Then \proc{base-kernel}(\id{C}, \id{A}, \id{B})
\li           \Return
        \End
\li \Comment Run the first $4$ sub-MMs concurrently
\li $\proc{CO2}(\id{C}_{00}, \id{A}_{00}, \id{B}_{00}) \parallel \proc{CO2}(\id{C}_{01}, \id{A}_{00}, \id{B}_{01})$
\li $\parallel \proc{CO2}(\id{C}_{10}, \id{A}_{10}, \id{B}_{00}) \parallel \proc{CO2}(\id{C}_{11}, \id{A}_{10}, \id{B}_{01})$
\li $\serial$ \Comment \Sync \label{li:mm-n2-sync}
\li \Comment Run the next $4$ sub-MMs concurrently
\li $\proc{CO2}(\id{C}_{00}, \id{A}_{01}, \id{B}_{10}) \parallel \proc{CO2}(\id{C}_{01}, \id{A}_{01}, \id{B}_{11})$
\li $\parallel \proc{CO2}(\id{C}_{10}, \id{A}_{11}, \id{B}_{10}) \parallel \proc{CO2}(\id{C}_{11}, \id{A}_{11}, \id{B}_{11})$
\li $\serial$ \Comment \Sync
\li \Return
% }% end resizebox
\end{codebox}
\vspace*{-1 \baselineskip}
\caption{Recursive MM algorithm with $O(n^2)$ space}
\label{fig:mm-n2}
\end{subfigure}
\caption{Recursive MM algorithms.
``$\parallel$'' and ``$\serial$'' are linguistic
constructs of Nested Dataflow model \cite{DinhSiTa16}
(\figref{symbols}, \secref{model}).
}
\label{fig:mm-dac}
\end{figure*}

% Briefly describe the theoretic models to calculate the bounds
% Move the description to the model section !!
% Calculate the work-time bounds of algorithms in 
% \figreftwo{mm-n3}{mm-n2} respectively and show the problem to 
% target in this paper
 
We can calculate the critical-path length (time), space and 
serial cache bounds 
of these two algorithms by the recurrences of 
\eqrefTwo{mm-n3-time}{madd-cache-stop}. 
% Some brief explanations on the recurrences!
\proc{CO2} algorithm in \figref{mm-n2} uses no more 
space than input and output matrices, thus its space
bound is simply $O(n^2)$.
\aeqref{mm-n3-space} says that \proc{CO3} algorithm allocates
an $n$-by-$n$ temporary matrix $D$ before spawning subtasks at
each recursion level, thus has an additional
overhead of matrix addition to merge results as indicated
by \eqref{mm-n3-time}.
\aeqref{mm-n2-time} shows that \proc{CO2} algorithm does not 
have this overhead. 
Because of the temporary matrix, \proc{CO3} algorithm can run
all eight ($8$) subtasks derived at each recursion
level completely in parallel, hence only one subtask sits on
the critical path as indicated
in \eqref{mm-n3-time}, while \proc{CO2} has to separate eight
subtasks into two parallel steps, i.e. two subtasks sitting
on its critical path as in \eqref{mm-n2-time}.
\aeqref{mm-n2-cache-stop} says that as soon
as the input and output matrices of \proc{CO2} 
are smaller than some constant fraction of cache size $M$, there 
will be no more cache misses than a serial scan.
\aeqref{mm-n3-cache-stop}, on the contrary, indicates that
\proc{CO3} algorithm keeps allocating new memory for new subtasks
, which is always assumed to incur cold cache misses.
\aeqreftwo{madd-time}{madd-cache} show that the matrix addition
is also done by a $2$-way divide-and-conquer and there are
four ($4$) subtasks derived at each recursion level. There is
no data dependency among subtasks thus only one sitting on
its critical path.
Solving the recurrences, we can see that \proc{CO3} 
algorithm in \figref{mm-n3} has an optimal 
$O(\log n)$ time bound (critical-path length) 
\footnote{The overhead of matrix
addition at each recursion level is just $O(1)$, i.e. 
$T_{\infty, \proc{madd}}(n) = O(1)$.} 
if counting only data dependency but a poor $O(n^3)$ space 
and $O(n^3/B)$ serial cache bounds; By contrast, \proc{CO2} 
algorithm 
has an optimal $O(n^2)$ space and $O(n^3/(B\sqrt{M}) + n^2/B)$
serial cache bound, but a sub-optimal $O(n)$ time bound 
(critical-path length).
In the literature, it is known as space-time tradeoff.
%
%\begin{figure}[!ht]
\begin{align}
T_{\infty, \proc{CO3}} (n) &= T_{\infty, \proc{CO3}} (n/2) + T_{\infty, \proc{madd}} (n) \label{eq:mm-n3-time} \\
S_{\proc{CO3}} (n) & = 8 S_{\proc{CO3}} (n/2) + n^2 \label{eq:mm-n3-space} \\
T_{\infty, \proc{madd}} (n) &= T_{\infty, \proc{madd}} (n/2) \label{eq:madd-time}\\
T_{\infty, \proc{CO2}} (n) &= 2 T_{\infty, \proc{CO2}} (n/2) \label{eq:mm-n2-time}\\ 
Q_{1, \proc{CO2}} (n) &= 8 Q_{1, \proc{CO2}} (n/2)\label{eq:mm-n2-cache} \\
Q_{1, \proc{CO2}} (n) &= O(n^2/B) \quad \quad \text{if $n \leq \epsilon M$} \label{eq:mm-n2-cache-stop} \\
Q_{1, \proc{CO3}} (n) &= 8 Q_{1, \proc{CO3}} (n/2) + Q_{1, \proc{madd}} (n)\label{eq:mm-n3-cache} \\
Q_{1, \proc{CO3}} (1) &= O(1) \label{eq:mm-n3-cache-stop} \\
Q_{1, \proc{madd}} (n) &= 4 Q_{1, \proc{madd}} (n/2) \label{eq:madd-cache} \\
Q_{1, \proc{madd}} (n) &= O(n^2/B) \quad \quad \text{if $n \leq \epsilon M$} \label{eq:madd-cache-stop} 
\end{align}
%\end{figure}
% 
An interesting research question is if it is possible 
to achieve a sublinear time bound while still bounding 
space and cache complexities to be asymptotically optimal.

% List of contribution bullets
\paragrf{Our Contributions (\figref{contri-table}): } 
\begin{itemize}
    \item We look into runtime system of PO algorithms and
        prove that if a runtime follows the ``busy-leaves''
        property \cite{BlumofeJoKu95}, there will be no more 
        than $p$ subtasks of the same depth co-exist at 
        any time in a $p$ processor system.
        If a runtime memory allocator stands by
        the ``Last-In First-Out'' principle, memory blocks 
        on the same processor are then largely
        reused, thus avoiding most of \proc{CO3} algorithm's
        cache misses without sacrificing critical-path length.
        Moreover, we propose a novel ``lazy allocation'' 
        strategy such that
        a subtask allocates temporary space ``\emph{after}''
        it is spawned and ``\emph{after}'' it makes sure that
        it is running simultaneously with its siblings that
        target the same output region.
        By the strategy, we reduce the space and
        cache requirement of PO MM algorithms to be 
        asymptotically optimal while still keeping a sublinear
        critical-path length.
        In \secref{strassen}, we show how to extend the
        approach to Strassen-like fast algorithms.
    \item We show by experiments that our new PO MM algorithms 
        do have performance advantage over both classic 
        \proc{CO2} and \proc{CO3} algorithms in a fair
        comparison.
\end{itemize}

\paragrf{Organization: }
\secref{model} introduces the cost models and programming
model for algorithm 
design and analysis; \secref{mm} discusses the intuitions behind 
our new PO MM algorithms, as well as its step-by-step 
construction; \secref{strassen}
extends our approach to Strassen-like fast MM algorithm;
\secref{expr} experiment and compare our new algorithms with
classic counterparts in a fair fashion;
\secref{relWork} concludes the paper and discusses related works.

\secput{model}{Cost Models and Programming Model}

This section briefly introduces the theoretical models used in 
algorithm design and analysis.
% Briefly describe the theoretic models to calculate the bounds

\paragrf{Parallel Performance Model: }
We adopt the work-time model \cite{JaJa92}
(also known as work-span model \cite{CormenLeRi09})
to calculate time complexities.
The model views a parallel computation as a DAG. 
Each vertex stands for a piece of computation 
that has no parallel construct and each edge represents some 
control or data dependency.
For simplicity, we count each arithmetic operation such as
multiplication, addition, and comparison uniformly as an 
$O(1)$ operation.
The model calculates an algorithm's time bound 
(also known as critical-path length, span, depth, denoted by
$T_\infty$) 
by counting the number of arithmetic operations along 
critical path. 
\punt{% begin punt
As pointed out in \cite{TangYoKa15, DinhSiTa16}, any extra 
control dependency to data dependency in an algorithm is 
artificial dependency and can only hurt critical path length,
hence should be eliminated by techniques such as Nested Dataflow 
(ND) model \cite{DinhSiTa16}. 
}% end punt
Work bound ($T_1$) is then the sum of all arithmetic operations.
Time bound $T_\infty$ and work bound $T_1$ characterize 
the running time of parallel algorithm on infinite number 
and one processor(s), respectively.
This paper assumes a Randomized Work-Stealing (RWS) scheduler 
that has the ``busy-leaves'' property as specified in 
\cite{BlumofeJoKu95}. More discussions on the property and its
application can be found in \secref{sar-mm}.
\punt{% begin punt
The busy-leaves property 
says that from the time a task is spawned to the time it 
finishes, there is always
at least one subtask from the subcomputation rooted at it 
that is ready. In other words, no leaf task can stall or be 
preempted. Busy-leaves property holds in Cilk
system \cite{IntelCilkPlus10, Leiserson10} and will be 
extended by our \thmref{remote-blocking}.
}% end punt
By an RWS scheduler, a better time bound, or equivalently
a shorter critical-path length, means more work available 
for randomized stealing
along critical path at runtime, hence a better ``dynamic
load-balance''.
We call a parallel algorithm \defn{work-efficient} if its total 
work $T_1$ matches asymptotically the time bound of best 
serial algorithm of the same problem.
Analogously, we have the notions of \defn{space-efficient} and 
\defn{cache-efficient}. 

\paragrf{Memory Model: }
We calculate only an algorithm's serial cache 
bound in the ideal cache model \cite{FrigoLePr12} since
corresponding parallel cache bound under RWS scheduler 
is bounded by \eqref{parallel-cache-bound} 
\cite{AcarBlBl00, SpoonhowerBlGi09}.  
\punt{% begin punt
\begin{align}
Q_p &= Q_1 + O(p T_\infty M/B) \label{eq:parallel-cache-bound}
\end{align}
}% end punt
% according to the ``drifted node / deviation'' method by Acar et 
% al 
Therefore, in the rest of paper, the term ``cache bound 
(complexity)'' stands for ``serial cache bound (complexity)''
unless otherwise specified.

The ideal cache model has an upper level cache of size $M$ and a
lower level memory of infinite size. 
Data exchange between the upper 
and lower level is coordinated by an omniscient (offline optimal) 
cache replacement algorithm in cache line of size $B$. It also 
assumes a tall cache, i.e. $M = \Omega(B^{2})$.
To accommodate parallel execution, 
we further assume that the lower level memory follows CREW 
(Concurrent Read Exclusive Write) convention.  Every concurrent 
reads from the same memory location can be accomplished in 
$O(1)$ time, while $n$ concurrent writes to the same memory cell 
have to be serialized by some order and take $O(n)$ total time 
to complete. That is to say, no matter these $n$ concurrent writes
are coordinated by user's atomic operation such as 
Compare-And-Swap (CAS) or by system's synchronization
facility such as ``\CilkSync'' in Cilk system, we always count 
their overall overhead by $O(n)$.
By \eqref{parallel-cache-bound}, $Q_p - Q_1 = O(p T_\infty M/B)$, 
we can see that the extra cache misses in a parallel execution 
to its serial execution is proportional to $T_\infty$, the 
critical-path length.
Hence, a shorter critical-path length means less parallel 
cache misses.

\paragrf{Programming Model: }
We use the notation ``$\id{a} \parallel \id{b}$'' to indicate 
that no subtasks of
\id{b} depend on any subtasks of \id{a}, i.e. \emph{no} 
dependency, while ``$\id{a} \serial \id{b}$'' says that all 
subtasks of \id{b} depend on all subtasks of \id{a}, i.e. a 
\emph{full} dependency. 

\secput{mm}{Space-Time Adaptive and Reductive (STAR) MM Algorithm}
% Organization of the section
\paragrf{Organization: } % 
\secref{tar-mm} parallelizes all multiplications of 
\proc{CO2} algorithm in \figref{mm-n2} without 
using much space; Based on \proc{CO3} 
algorithm in \figref{mm-n3}, \secref{sar-mm} 
reduces space requirement by exploiting the ``busy-leaves''
property \cite{BlumofeJoKu95}, bounds cache complexity to
be asymptotically optimal
by requiring a ``Last-In First-Out'' memory allocator, and
further improves space and cache requirements by ``lazy 
allocation''; \secref{star-mm} bounds the space complexity
to be asymptotically optimal by one level of indirection of 
TAR and SAR based on processor count $p$.

\begin{figure}
\begin{subfigure}[b]{0.5 \textwidth}
%\scriptsize
\begin{codebox}
\Procname{\proc{TAR-MM}(\id{C}, \id{A}, \id{B})}
\li \Comment $\id{C} \assign \id{A} \times \id{B}$
\li \If $(\func{sizeof}(\id{C}) \leq \const{base\_size})$
\li     \Then \Comment Request space from the program-managed memory pool
\li           $\id{D} \assign \proc{get-storage}(\func{sizeof}(C))$ \label{li:get-storage}
\li           \proc{base-kernel}(\id{D}, \id{A}, \id{B})
%\li           \While (there is concurrent writes on \id{C}); \label{li:tar-mm-block}
\li           \Comment Write the intermediate results in \id{D} to \id{C} atomically
\li           \proc{atomic-madd}(\id{C}, \id{D}) \label{li:tar-mm-atomic-write}
\li           \Comment Return storage to the memory pool
\li           \func{free}(\id{D})
\li           \Return
        \End
\li \Comment Run all $8$ sub-MMs concurrently
\li $\proc{TAR-MM}(\id{C}_{00}, \id{A}_{00}, \id{B}_{00}) \parallel \proc{TAR-MM}(\id{C}_{01}, \id{A}_{00}, \id{B}_{01})$
\li $\parallel \proc{TAR-MM}(\id{C}_{10}, \id{A}_{10}, \id{B}_{00}) \parallel \proc{TAR-MM}(\id{C}_{11}, \id{A}_{10}, \id{B}_{01})$
\li $\parallel \proc{TAR-MM}(\id{C}_{00}, \id{A}_{01}, \id{B}_{10}) \parallel \proc{TAR-MM}(\id{C}_{01}, \id{A}_{01}, \id{B}_{11})$
\li $\parallel \proc{TAR-MM}(\id{C}_{10}, \id{A}_{11}, \id{B}_{10}) \parallel \proc{TAR-MM}(\id{C}_{11}, \id{A}_{11}, \id{B}_{11})$
\li \Return
\end{codebox}
\vspace*{-1 \baselineskip}
\caption{\proc{TAR-MM} algorithm}
\label{fig:tar-mm}
\end{subfigure}
\vfil
%
%\begin{minipage}[b]{0.5 \textwidth}
\begin{subfigure}[b]{0.5 \textwidth}
%\scriptsize
\begin{codebox}
\Procname{\proc{hlp}(\id{Parent}, \id{A}, \id{B}, \id{d})}
\li \If $(\id{parent}.\func{trylock}())$ \label{li:hlp-get-storage-begin}
\li     \Then \Comment work right on parent's storage
\li         $\id{D} \assign \id{parent}$
\li     \Else
\li         \Comment request space for depth $d$ 
\li         $\id{D} \assign \proc{get-storage}(\func{sizeof}(n/2^{\id{d}}))$ \label{li:wrapper-sar-get-storage}
        \End \label{li:hlp-get-storage-end}
\li \If $(\func{sizeof}(n/2^{\id{d}}) \leq \const{base\_size})$
\li     \Then \proc{base-kernel}(\id{D}, \id{A}, \id{B})
\li     \Else
\li         \proc{SAR-MM}(\id{D}, \id{A}, \id{B}, \id{d})
    \End
\li \If $(\id{D} \neq \id{parent})$ \label{li:hlp-wb-begin}
\li     \Then \Comment Update \id{D} to \id{parent} atomically
\li         \proc{atomic-madd}(\id{parent}, \id{D}) \label{li:wrapper-sar-atomic-madd}
\li         \Comment Return storage to the memory pool
\li         \func{free}(\id{D})
\li     \Else
\li         $\id{parent}.\func{unlock}()$ \label{li:hlp-unlock}
        \End \label{li:hlp-wb-end}
\li \Return
\end{codebox}
\vspace*{-1 \baselineskip}
\caption{The helper function request temporary storage from the
program-managed memory pool \emph{iff} parent's storage is
occupied. If computation is on a local temporary
storage, the helper function will write back results to
parent by atomic addition.}
\label{fig:wrapper-sar-mm}
\end{subfigure}
\vfil
\begin{subfigure}[b]{0.5 \textwidth}
%\scriptsize
\begin{codebox}
\Procname{\proc{SAR-MM}(\id{C}, \id{A}, \id{B}, \id{d})}
\li \Comment Computes SAR-MM at recursion level \id{d}
\li \Comment Run all $8$ sub-MMs concurrently
\li $\proc{hlp}(\id{C}_{00}, \id{A}_{00}, \id{B}_{00}, \id{d} + 1) \parallel \proc{hlp}(\id{C}_{01}, \id{A}_{00}, \id{B}_{01}, \id{d} + 1)$
\li $\parallel \proc{hlp}(\id{C}_{10}, \id{A}_{10}, \id{B}_{00}, \id{d} + 1) \parallel \proc{hlp}(\id{C}_{11}, \id{A}_{10}, \id{B}_{01}, \id{d} + 1)$
\li $\parallel \proc{hlp}(\id{C}_{00}, \id{A}_{01}, \id{B}_{10}, \id{d} + 1) \parallel \proc{hlp}(\id{C}_{01}, \id{A}_{01}, \id{B}_{11}, \id{d} + 1)$
\li $\parallel \proc{hlp}(\id{C}_{10}, \id{A}_{11}, \id{B}_{10}, \id{d} + 1) \parallel \proc{hlp}(\id{C}_{11}, \id{A}_{11}, \id{B}_{11}, \id{d} + 1)$
\li \Return
\end{codebox}
\vspace*{-1 \baselineskip}
\caption{\proc{SAR-MM} algorithm}
\label{fig:sar-mm}
\end{subfigure}
%\end{minipage}
\caption{\proc{TAR-MM} and \proc{SAR-MM} algorithms.
}
\label{fig:mm-dac}
\end{figure}

\subsecput{tar-mm}{Time Adaptive and Reductive (TAR) MM Algorithm} 
A close look at \proc{CO2} algorithm in \figref{mm-n2}
reveals several aspects for further improvement. 
\begin{enumerate}
\item It imposes more control dependency than necessary data 
    dependency to keep the algorithm correct.
    That is, the all-to-all synchronization on \liref{mm-n2-sync} 
    of \figref{mm-n2} serializes eight sub-MM's of each recursion 
    levels to two parallel steps.
    For an instance, by this synchronization, the computation of 
    \proc{CO2}$(\id{C}_{00}, \id{A}_{01}, \id{B}_{10})$ 
    not only waits on 
    \proc{CO2}$(\id{C}_{00}, \id{A}_{00}, \id{B}_{00})$
    , which writes to the same $\id{C}_{00}$ quadrant, but also 
    has to wait for the computation working on completely disjoint
    quadrants, i.e. $\id{C}_{01}$, $\id{C}_{10}$, and 
    $\id{C}_{11}$.

\item The synchronization on \liref{mm-n2-sync} of \figref{mm-n2} 
    essentially serializes $n$ multiplications updating
    the same cell of output matrix to $n$ parallel steps. 
    However, multiplications by 
    themselves do not have any data dependency among each other 
    and should be parallelized. The serialization makes 
    sense only on later writing back by additions.
\end{enumerate}

Based on the above observations, we devise a Time Adaptive
and Reductive (TAR) algorithm to remove unnecessary 
control dependency from critical path. \afigref{tar-mm} shows the 
pseudo-code.
\proc{TAR-MM} algorithm spawns all eight sub-MM's at each level of 
recursion simultaneously to maximize parallelism and 
serialize only concurrent writes to the same output region.
Though we employ atomic operation such as 
\func{Compare-And-Swap} (\proc{atomic-madd} on 
\liref{tar-mm-atomic-write}) in our pseudo-code for the 
serialization, we feel that it's possible to design a 
dataflow operator like the ``$\fire$'' operator in Nested Dataflow 
model \cite{DinhSiTa16} or ``\CilkSync'' in Cilk system for 
the purpose. 
When output region \id{C} is of \const{base\_size}, the algorithm
requests temporary storage from memory allocator
(\liref{get-storage}) before base-case computation. 

\paragrf{Memory Allocator: }
Memory allocator is a key component to guarantee the reuse of 
data blocks across requests on each processor. 
The reason that \proc{CO3} algorithm incurs $O(n^3/B)$ cache
misses, which is asymptotically more than that of \proc{CO2},
i.e. $O(n^3/(B\sqrt{M}))$, is because it repeatedly requests
space before spawning subtasks at each depth, and
people assume that allocation of space will always incur cold
cache misses to fill it.
If a memory allocator can guarantee the reuse of memory block,
the above assumption is no longer true, thus we can save lots
of un-necessary cold misses. That is to say, though \proc{CO3}
algorithm still requests space at each depth, if
the space is reused from prior requests, the fill-out of 
space by new data will not incur any cold misses because a smart
cache can hold the reused memory block in cache by the omniscient
cache replacement policy \cite{FrigoLePr12}.
More precisely,
all requests on the same processor should be served in a Last-In, 
First-Out (LIFO) fashion like a stack so that a smart cache can 
hold most recently used data blocks in cache 
to avoid thrashing. More precisely, if a user's program 
requests the same sized memory block on the same processor, 
allocator should guarantee to return exactly the same memory
block for reuse.
\punt{% begin punt
Experimental results in \secref{expr} show that our 
program-managed 
memory pool easily outperforms Intel TBB's Thread-Local Storage 
(TLS) \cite{Reinders07}. 
}% end punt

We have an observation that each processor can work on only one 
task (one nested function call in an invocation tree)
at a time. We further assume that a task 
computing a base case can not block or be preempted (according to 
the ``busy-leaves'' property \cite{BlumofeJoKu95}), which is true 
by Cilk's RWS scheduler \cite{BlumofeLe99}. 
Then we have \thmref{tar-mm}.

\begin{theorem}
    There is a \proc{TAR-MM} algorithm that computes general 
    square MM of dimension $n$ on a semiring in $O(n)$ time, 
    $O(n^2 + p b^2)$ space, and optimal $O(n^3/(B\sqrt{M}) + 
    n^2/B)$ 
    cache misses, where $b$ denotes the dimension of base case.
    If assuming $b$ is some small constant, the space bound
    reduces to $O(n^2 + p)$.
\label{thm:tar-mm}
\end{theorem}
\begin{proof} %
\paragrf{Time bound: }
The critical-path length of $O(n)$ follows from the fact that
the algorithm parallelizes all multiplications 
and serializes only concurrent writes to the same memory 
location. 
\punt{% begin punt
Recall from the paragraph of ``Memory Model'' in \secref{model} 
that $O(n)$ writes to the same memory location are always 
serialized in $O(n)$ time.
}% end punt

\paragrf{Space bound: }
Space bound follows from the facts that each processor can 
work on only one base case at a time and base-case computation 
can not block or be preempted. The temporary space for base-case 
computation on each processor thus are reused across different
invocations.

\paragrf{Cache bound: }
The recurrences for cache bound are almost 
identical to that of \proc{CO2} except that the stop 
condition is changed to \eqref{tar-mm-cache-stop} as follows.
\begin{align}
    Q_{1, \proc{TAR-MM}} (n) &= 8 Q_{1, \proc{TAR-MM}} (n/2) & & \label{eq:tar-mm-cache} \\
    Q_{1, \proc{TAR-MM}} (n) &= O(n^2/B) & & \text{if $3 n^2 + b^2 \leq \epsilon M$} \label{eq:tar-mm-cache-stop} 
\end{align}
\aeqref{tar-mm-cache} recursively calculates a dimension-$n$
(an $n$-by-$n$ matrix multiplies another $n$-by-$n$ matrix)
\proc{TAR-MM}'s cache misses to eight ($8$) dimension-$(n/2)$
\proc{TAR-MM}'s cache misses.
When the space requirements of input, output and temporary storage 
(a $b$-by-$b$ storage allocated on \liref{get-storage} of 
\figref{tar-mm})
of a dimension-$n$ \proc{TAR-MM} are less than or equal some 
constant factor of $M$, any further recursion will not incur
more cache misses than a serial scan as addressed by 
\eqref{tar-mm-cache-stop}. Solving the recurrences will yield
the bound.
\punt{% begin punt
Note that the ideal cache model \cite{FrigoLePr12} assumes an 
omniscient cache 
replacement policy that will always hold the $b$-by-$b$ 
storage in cache to avoid thrashing.
The theorem then follows.
}% end punt
\end{proof}

\punt{% begin punt -- SODA doesn't care much about performance and non-asymptotic optimization 
\paragrf{Optimization: } There is one possible optimization for the \proc{TAR-MM} to 
reduce the number of updates to the global output matrix, which requires atomic operations
, and reduce the conflicts of concurrent writes.
Each processor will have $p$, instead of $1$, copies of base case. 
When comes to a base case computation, the processor will first
check to see if the base case is already held in one of the $p$ local copies or if there
is any free local storage. If so, it computes the multiply-add operations of the
base case on local copies; If all $p$ local copies are occupied and the 
computing base case is not one of them, it will flush all $p$ copies of
intermediate results to the global output matrix \id{C} before it computes the base
case locally.
Since there are at most $p$ possible conflicts of concurrent writes from different
processors, the processor can simply switch to the update of other local copies
if any conflict is detected. 
}% end punt
\punt{% begin punt
Experimental results in \secref{expr} show that the 
\proc{TAR-MM} algorithm
consistently outperforms \proc{CO2} algorithm 
(\figref{mm-n2}) by
$5-10\%$ when the problem size is reasonably large.
}% end punt

\subsecput{sar-mm}{Space Adaptive and Reductive (SAR) MM algorithm}
A close look at \proc{CO3} algorithm (\figref{mm-n3})
shows that
it is designed for system with infinite or sufficient number 
of processors (proportional to algorithm's parallelism of 
$T_1/T_\infty$). 
At each level of recursion, regardless of availability of 
idle processors, the algorithm always allocates a temporary 
matrix \id{D} of the same size as output matrix \id{C}
(\liref{mm-n3-alloc-D} in \figref{mm-n3}). 
By recursion, it allocates $n^3 - n^2$ total temporary space 
on a $p$-processor system
, where $p \ll T_1/T_\infty = O(n^3 / \log n)$ usually holds 
in reality.

To reduce space requirement at no cost of parallelism,
we have one observation and one algorithmic trick as follows.
\punt{% begin punt
The observation is that there can be no more than $p$ computing
tasks, i.e. nested function calls, of the same depth
executing or blocking at any time by the ``busy-leaves'' property 
as addressed by \thmref{general-busy-leaves}.
The algorithmic trick is instead of allocating temporary space 
before spawning recursive function calls, we allocate space in 
a lazy fashion (\figref{wrapper-sar-mm}).
}% end punt
\punt{% begin punt
We have an observation that each processor at one time can work 
on only one computing task, i.e. one nested function call. 
}% end punt
\paragrf{Generalization of ``busy-leaves'' property: }
If we view the 
execution of a recursive algorithm as a depth-first traversal 
of its invocation tree, each node of which stands for a 
computing task (nested function call),
we define the \defn{depth} of a node to be the number of nodes 
on the path from root of tree to itself. 
The RWS scheduler specified in Blumofe and Leiserson's paper 
\cite{BlumofeLe99} has an important \defn{busy-leaves}
property as follows. 
The busy-leaves property says that from the time a task is 
spawned to the time it finishes, there is always
at least one subtask from the subcomputation rooted at it
that is ready. In other words, no leaf task can stall or be
preempted.  The busy-leaves property holds in Cilk runtime
system \cite{IntelCilkPlus10, Leiserson10} and we further 
extend it by \thmref{general-busy-leaves}.

\begin{theorem}
If a runtime scheduler stands by the \defn{busy-leaves} property, 
there can be no more than $p$ tasks of the same depth executing 
or blocking at any time in a $p$-processor system.
\label{thm:general-busy-leaves}
\end{theorem}
\begin{proof}
We prove by induction on the depth of tasks.
Since no leaf task can stall or be preempted, there can be no
more than $p$ leaf tasks at any time in a $p$-processor system.
Since each leaf task can have only one parent task, it's obvious
that their parent tasks can be no more than $p$ either. 
Recursively, the argument holds for arbitrary depth \id{d}.
\punt{% begin punt
    % this old proof can be substituted by new one
There is only one depth-$0$ task, i.e. root task, so the 
conclusion vacuously holds. Assuming that the conclusion holds 
up to 
depth \id{d}, we prove that it will hold for depth $\id{d}+1$. 
When a depth-\id{d} task spawns a depth-$(\id{d}+1)$ child, if 
the runtime scheduler chooses to run the child first, the 
depth-\id{d} parent will be pushed into the bottom end
of local deque. It must be resumed to execute before it can spawn 
another depth-$(\id{d}+1)$ child. If it gets executed on the 
same processor, it must be that the previous depth-$(\id{d}+1)$ 
child has finished according to busy-leaves property. 
In this case, any new depth-$(\id{d}+1)$ child from the same 
parent runs strictly after their older depth-$(\id{d}+1)$ 
siblings. If the depth-\id{d} parent is stolen and executed on 
a different processor before its local depth-$(\id{d}+1)$ child
finishes, the total number of depth-$(\id{d}+1)$ tasks
is then bounded by the total number of processors, i.e. $p$. 
Hence, the conclusion holds in both cases.
If the scheduler chooses to run the parent first at a spawn, 
it will push the child into local deque. Since the 
depth-$(\id{d}+1)$ child is in deque, neither  
executing nor blocking, the conclusion trivially holds. By 
blocking or being preempted, we mean that a computing task is 
put aside somewhere in memory after it starts executing but 
before it finishes. If the child is then stolen and executed
on another processor, again the total number of non-blocking
depth-$(\id{d}+1)$ child is bounded by $p$.
This finalizes the proof for depth $\id{d}+1$.
}% end punt
\end{proof}

We verified by experiments that \thmref{general-busy-leaves} does
hold in Intel Cilk Plus runtime \cite{IntelCilkPlus10}.
By \thmref{general-busy-leaves}, it is sufficient to
allocate at most $p$ copies of sub-matrix of any
depth for reuse among all subtasks. More precisely, 
$\min\{p, 4^d\}$ copies of sub-matrices of any depth 
$d$ are sufficient, where the term $4^d$ indicates that at each 
depth, at most four out of eight sub-MMs will require 
temporary space and another four will work right on its 
parent's space. 
So, for any depth $d$, the memory allocator needs to hold 
at most $\min\{p, 4^d\}$ blocks of size $n/2^d$-by-$n/2^d$
for reuse across requests. 

\paragrf{Final SAR algorithm: }
The pseudo-code of this new algorithm, which we call 
\proc{SAR-MM}, is shown in \figref{sar-mm}, with a helper 
function in \figref{wrapper-sar-mm}.
%
% Explains the pseudo-code of algorithm ...

\paragrf{Lazy Allocation: }
To minimize space requirement, or in other words maximize
space reuse, we use an algorithmic trick
that allocates temporary space in a lazy fashion. That is, a 
sub-MM will request for temporary space if and only if it 
runs simultaneously on a different processor from
the sub-MM updating the same output region. 
In \figref{sar-mm}, all top-half and bottom-half sub-MMs
updating the same output regions are spawned simultaneously.
In \figref{wrapper-sar-mm}, \liref{hlp-get-storage-begin}
show that the top-half and corresponding bottom-half will 
compete on a mutex lock
to determine who will reuse parent's space and who will request
temporary space. If the
top-half and bottom-half are executed one-after-another either
on the same processor or on different processor, they will
both reuse parent's space.
If some sub-MM does request temporary space for local computation 
, it has to write its results back to parent 
atomically on \liref{wrapper-sar-atomic-madd}. 
Though we utilize atomic operations such as mutex lock to 
coordinate between every pairs of top-half and bottom-half,
we feel that it is possible and will be beneficial to have a 
system's facility like a dataflow ``$\fire$'' operator 
\cite{DinhSiTa16} for the purpose.
Besides the way shown in \figref{wrapper-sar-mm}, 
an alternative way to coordinate space reuse is to always let
the bottom-half reuse parent's space and top-half check the
status (``running'' or ``finished'') of corresponding bottom-half 
when it is scheduled to run before it decides if it should 
request temporary space or just reuse bottom-half's space.
We have to clarify that \thmref{general-busy-leaves} always
holds with or without the algorithmic trick of lazy allocation.
The trick just further maximizes the space reuse, thus reduces 
cache misses.

\begin{theorem}
    There is a \proc{SAR-MM} algorithm that computes general
    MM of dimension $n$ on a semiring in optimal $O(\log n)$ time,
    $O(p^{1/3} n^2)$ space, and optimal $O(n^3/B\sqrt{M} + 
    n^2/B)$ cache bounds, assuming $p = o(n)$.
\label{thm:sar-mm}
\end{theorem}
\begin{proof}
\paragrf{Time bound: }
Since the algorithm does not impose any synchronization among
the eight ($8$) sub-MMs derived at each depth, the
time bound is not affected, i.e. there are still 
$O(\log n)$ (atomic) additions sitting on critical path.
We have to clarify that the \proc{atomic-madd} on 
\liref{wrapper-sar-atomic-madd} in \figref{wrapper-sar-mm}
is functionally equivalent to a summation of the syncrhonization 
on \liref{co3-sync} of \proc{CO3} algorithm in \figref{mm-n3}
and the \func{madd} on \liref{co3-madd}. We just implement
it by atomic operation such as Compare-And-Swap (\func{CAS}).
Moreover, the \func{trylock} on \liref{hlp-get-storage-begin}
and \func{unlock} on \liref{hlp-unlock} of 
\figref{wrapper-sar-mm} is just an $O(1)$
operation because every pair of top-half and bottom-half do 
not wait on each other by the operation, but just check and 
signal the status to each other.

\paragrf{Space bound: }
The recurrences of temporary space requirements are:
\begin{align}
    & S (v) = 8 S (v/2) + 4 (v/2)^2 & \text{if $v > n/2^{\id{k}}$} \label{eq:sar-mm-upper}\\ 
    & S (v) = p S_1 (v) & \text{if $v \leq n/2^{\id{k}}$} \label{eq:sar-mm-bottom-p}\\
    & S_1 (v) = S_1 (v/2) + (v/2)^2 & \text{if $v \leq n/2^{\id{k}}$} \label{eq:sar-mm-bottom}\\
    & 4 \times (8^0 + 8^1 + \ldots + 8^{\id{k}}) = p & \label{eq:sar-mm-depth}
\end{align}
The term $S_1$ stands for the space requirement on each 
processor. \aeqref{sar-mm-upper} says that at upper levels 
of recursion, every four out of eight sub-MMs spawned at
each level may require extra space. 
\aeqref{sar-mm-depth} calculates the switching depth of \id{k}.
% \punt{ % begin punt
The term $(n/2)^2$ on the right-hand side of 
\eqref{sar-mm-bottom} 
indicates that only one copy of $n/2$-by-$n/2$ temporary 
matrix is needed for all MMs of size $n$-by-$n$ on any single 
processor.
% } % end punt
%
The recurrences solve to $\id{k} = (1/3) \log_2 (7/8 p + 1/2)$ 
and $S(n) = O(p^{1/3} n^2)$. Assuming $p^{1/3} \ll n$, i.e. 
$p = o(n)$, which 
usually holds in reality, the total space bound can be 
asymptotically less than that of the \proc{CO3} algorithm. 

\paragrf{Cache bound: }
The recurrences of cache complexity are:
\begin{align}
&Q_1 (n) = 8 Q_1 (n/2) + O(n^2/B) & \label{eq:sar-mm-cache}\\
&Q_1 (n) = O(n^2/B) \quad \text{if $(1 + 1/4 + \ldots) n^2 + 2 n^2 \leq \epsilon_4 M$} & \label{eq:sar-mm-cache-stop}
\end{align}
\aeqref{sar-mm-cache} is identical to that of the \proc{CO3}
algorithm, where the $O(n^2/B)$ term accounts for the possible 
overheads of merging the top-half and corresponding bottom-half's 
results by addition. If the top-half reuses
bottom-half's space by lazy allocation, the overhead will 
disappear. So \eqref{sar-mm-cache} accounts for the worst case.
The stop condition of \eqref{sar-mm-cache-stop} says that if the 
summation of memory footprint of all writes (output region) and 
reads (input region) of later recursions of dimension $n$ are 
less than or equal some constant fraction 
of cache size $M$, it will incur no more cache misses than a 
serial scan. In the equation, the term $(1 + 1/4 + \ldots) n^2$ 
stands for the summation of all subsequent writes' memory 
footprint assuming that same-sized memory requests reuse the
same memory blocks on the same processor, and $2 n^2$ is of 
all reads. The recurrences solve to the optimal 
$O(n^3/(B\sqrt{M}) + n^2/B)$ bound.
\end{proof}

\subsecput{star-mm}{Space-Time Adaptive and Reductive (STAR) 
MM algorithm}
The TAR algorithm removes multiplications from 
critical path without using much more space, while the SAR 
algorithm reduces space requirement without
increasing time bound. A natural idea
will be combining the two and yields
a near time-optimal and space-optimal \defn{STAR} MM algorithm.

The hybrid algorithm works as follows.
At upper levels of recursion, we employ the TAR algorithm and 
switch to the SAR after depth \id{k}, where 
\id{k} is a parameter to be determined later. 
More precisely, before 
depth \id{k}, our new algorithm spawns all eight sub-MMs at 
each depth in parallel and serializes 
concurrent writes to the same output region like the 
\proc{TAR-MM} algorithm;
After depth \id{k}, the algorithm keeps spawning all eight 
sub-MMs at each depth in parallel and reuse memory 
blocks like the \proc{SAR-MM}. 

\begin{theorem}
    There is a \proc{STAR-MM} algorithm that computes the 
    general MM of dimension $n$ on a semiring in 
    $O(\sqrt{p} \log n)$ time, optimal 
    $O(n^2)$ space, and optimal $O(n^3/(B \sqrt{M}) + n^2/B)$ 
    cache bounds, assuming $p = o(n^2/\log^2 n)$.
\label{thm:star-mm}
\end{theorem}
\begin{proof}
The time and space recurrences of the STAR MM algorithm are:
\begin{align}
    T_\infty (v) &= 2 T_\infty (v/2) & \text{if $v > n/2^\id{k}$}  \label{eq:star-mm-upper-time}\\
    T_\infty (v) &= T_\infty (v/2) + O(1) & \text{if $v \leq n/2^\id{k}$}  \label{eq:star-mm-lower-time}\\ 
    S (v) &= p S_1 (v) & \text{if $v \leq n/2^{\id{k}}$}  \label{eq:star-mm-space}\\
S_1 (v) &= S_1 (v/2) + (v/2)^2 & {}  \label{eq:star-mm-single-space}
\end{align}
\aeqreftwo{star-mm-upper-time}{star-mm-lower-time} indicate 
that concurrent writes to the same output region 
are serialized / parallelized before / after depth \id{k}
, respectively. Since there are at most two
sub-MMs at each depth updating the same output 
quadrant, the serialization overhead is $O(1)$ at each level. 
Since no temporary space is requested before 
depth \id{k}, \eqref{star-mm-space} counts 
only the space requirement after \id{k},
which has the same form as in the SAR algorithm.
The recurrences solve to $S (v) = (1/3) p (n/2^\id{k})^2$ , and 
$T_\infty (n) = 2^\id{k} \log_2 (n/2^\id{k})$. Making 
$\id{k} = (1/2) \log_2 p$, we have $S(v) = (1/3) n^2 = O(n^2)$ 
and $T_\infty (n) = \sqrt{p} \log_2 (n/\sqrt{p}) = O(\sqrt{p} 
\log n)$.

We can consider the cache bound as follows. Recall that from
the paragraph of ``Memory Model'' in \secref{model} we count
only serial cache bound in this paper.
If executing STAR algorithm on a single processor, i.e. $p = 1$,
it reduces to SAR; If we adjust the switching depth \id{k} on
the single processor, it will become some middle ground between
the TAR and SAR. In either case, its cache bound stays 
optimal with a constant in big-Oh between that of the TAR and SAR.
\end{proof}

From the proof of \thmref{star-mm}, we can see that the total 
extra space requirement of STAR algorithm is just a third 
of the output matrix size, i.e. $(1/3) n^2$, with an 
$O(\sqrt{p})$ factor increase in the time bound. 
Since the \proc{CO2} algorithm 
uses at least $3 n^2$ space to hold the input and output matrices, 
this extra temporary space is just minimum. 
If we assume that $p = o(n^2/\log^2 n)$, the time bound stays 
sublinear.

\subsecput{discussion}{Discussions}
\punt{% begin punt
Our \proc{TAR-MM} algorithm removes un-necessary data dependency
among multiplications without using much space, and \proc{SAR-MM}
further reduces serial cache requirement to be asymptotically 
optimal by exploring runtime scheduler's ``busy-leaves'' property
memory allocator's ``Last-In First-Out'' property, and our novel
``lazy allocation'' property.
Finally, the \proc{STAR-MM} algorithm combines the advantages
of TAR and SAR by using processor count $p$ to bound space
requirement to be asymptotically optimal at a little increase of
critical-path length.
}% end punt
The main difference of our \proc{SAR-MM} algorithm from 
\proc{CO3} is that \proc{SAR-MM} requires a memory allocator 
to guarantee
the reuse of the same sized memory blocks across different
requests on the same processor, plus a novel ``lazy allocation'' 
strategy.
The reuse of memory blocks across requests removes un-necessary
cold cache misses in \proc{CO3} algorithm.
The generalization of ``busy-leaves'' property bounds total
space.
The ``lazy allocation'' trick further reduces space and 
cache requirements.

Though our STAR algorithm takes processor count $p$ as a 
parameter to bound total space requirement, unlike 
classic processor-aware approach, we do not partition statically 
problem space to processor grid. Static partitioning strategy 
has several concerns as discussed in \secref{intro}. 
By contrast, STAR algorithm enjoys the full benefits of
\emph{dynamic load balance} as classic processor-oblivious 
approach, which is a key to a satisfactory speedup especially 
when there is no good static partitioning algorithm for a 
problem. 
% So we call our design strategy ``\defn{Processor-Adaptive}''.
\punt{% begin punt
In summary, we conclude that our \proc{STAR-MM} algorithm has 
a sublinear and near-optimal time bound if processor count is 
relatively small to problem dimension, while staying work-, 
space- and cache-efficient. 
}% end punt
By a sublinear time bound, i.e. a 
shorter critical-path length of computational DAG,
our STAR algorithm has at least two \textbf{\textit{advantages}} 
over classic cache-oblivious parallel counterpart 
(e.g. the \proc{CO2} algorithm in \figref{mm-n2} 
). 
Firstly, it means more work available for dynamic 
load balance along critical path; 
Secondly, it incurs asymptotically less parallel cache misses 
according to \eqref{parallel-cache-bound}.

\secput{strassen}{STAR algorithm for Strassen-like fast matrix 
multiplication algorithm}

\punt{% begin punt
% Why we need Strassen-like algorithm
The expected running time of a parallel algorithm on a 
$p$-processor system under the RWS scheduler
is $T_1/p + O(T_\infty)$ \cite{BlumofeLe99}.
All MM algorithms discussed in \secref{mm} have 
$T_1 (n) = O(n^3)$. If assuming $p = o(n^2)$,
$T_1 (n) / p \gg T_\infty (n)$ hold for those MM algorithms.
This inequality says that decreasing the $T_\infty$ alone won't 
help much on the expected running time. We need decreasing the
total work $T_1$ as well.
}% end punt

This section extends the STAR technique to Strassen-like fast
MM algorithms.
Given square matrices $A$, $B$, and $C$, the Strassen algorithm 
\cite{Strassen69} recursively divides each matrix into four 
equally sized quadrants as shown in \eqref{mm-dac}. It then 
computes each quadrant of $C$ as follows:
\begin{align*}
S_1 &= A_{00} \oplus A_{11} & S_2 &= A_{10} \oplus A_{11} & S_3 &= A_{00} \\
S_4 &= A_{11} & S_5 &= A_{00} \oplus A_{01} & S_6 &= A_{10} \ominus A_{00} \\
S_7 &= A_{01} \ominus A_{11} & & & & \\
T_1 &= B_{00} \oplus B_{11} & T_2 &= B_{00} & T_3 &= B_{01} \ominus B_{11} \\
T_4 &= B_{10} \ominus B_{00} & T_5 &= B_{11} & T_6 &= B_{00} \oplus B_{01} \\
T_7 &= B_{10} \oplus B_{11} & & & & 
\end{align*}
\begin{align*}
P_r &= S_r \otimes T_r, & 1 \leq r \leq 7 
\end{align*}
\begin{align*}
C_{00} &= P_1 \oplus P_4 \ominus P_5 \oplus P_7 & C_{01} &= P_3 \oplus P_5 \\
C_{10} &= P_2 \oplus P_4 & C_{11} &= P_1 \oplus P_3 \ominus P_2 \oplus P_6
\end{align*}
The $\ominus$ notation dnotes the inverse operation of $\oplus$. 
\punt{% begin punt
Since not all $\oplus$ operations of a semiring has an inverse,
Strassen algorithm doesn't work for a general semiring, which is
the reason that it does not apply to the DP computation
discussed in the paper.
}% end punt

\begin{lemma}
    A straightforward parallelization of Strassen MM algorithm 
    has an $O(\log n)$ time, $O(n^{\log_2 7})$ work, 
    $O(n^{\log_2 7})$ space, and an $O(n^{\log_2 7}/B)$ serial 
    cache bound.
\label{lem:par-strassen}
\end{lemma}
% \punt{% begin punt
\begin{proof}
The time ($T_\infty$), space ($S$) and cache ($Q_1$) recurrences 
of a straightforward Strassen parallelization is as follows:
\begin{align*}
T_\infty (n) &= T_\infty (n/2) + O(1) & S (n) &= 7 S (n/2) + 17 (n/2)^2 \\
Q_1 (n) &= 7 Q_1 (n/2) + O(n^2/B) &  
\end{align*}
The time recurrence says that at each depth of 
dimension $n$, the Strassen algorithm spawns simultaneously 
seven ($7$) subtasks of dimension $n/2$ and only one of them 
sits on the critical path.
The $O(1)$ term in the time recurrence accounts for the overheads 
of constant number of matrix additions and subtractions (for 
computing $S$'s, $T$'s, and $C$'s) at each depth.
The space recurrence says that at each depth of
dimension $n$, a straightforward parallelization requires at most 
$17$ copies of $n/2$-by-$n/2$ temporary matrices to hold the 
intermediate results of all $P_r$'s and some of the $S_r$'s 
and $T_r$'s. If an $S_r$ or $T_r$ corresponds directly to 
a quadrant of input matrices \id{A} or \id{B} with no 
$\oplus$ or $\ominus$ operations such as $S_3$ or $T_2$, 
the algorithm doesn't allocate temporary space for it. 
Since all seven subtasks of multiplication of the same recursion 
level execute simultaneously, all temporary matrices have to be
ready before subtasks can be launched as the \proc{CO3} algorithm.
The cache recurrence is analogous with a similar stop 
condition to that of the \proc{CO3} algorithm 
(\eqref{mm-n3-cache-stop}), i.e. $Q_1 (1) = O(1)$, 
which indicates that the straightforward parallelization 
keeps allocating new space for new subtasks, which is 
always assumed to incur cold cache misses.
The recurrences solve to $T_\infty (n) = O(\log n)$,
$S(n) = O(n^{\log_2 7})$, and $Q_1 (n) = O(n^{\log_2 7}/B)$. 
That is to say, the amount of space requirement, as well as 
cache misses, is proportional to total work.
\end{proof}
% }% end punt

\punt{ % begin punt
\subsecput{tar-strassen}{TAR algorithm for Strassen}
The TAR algorithm is based on an algorithm that has no 
concurrent writes on the resulting matrix $C$. So we first
devise such an algorithm then uses the TAR technique to remove
artificial dependencies. 

In the original Strassen algorithm, each quadrant of $C$ comes
from the additions ($\oplus$) and subtractions ($\ominus$)
of several different $P_r$'s, i.e. concurrent writes. That's 
the reason that the algorithm requires
temporary matrices to hold all $P_r$'s. To remove the space for 
$P_r$'s without introducing conflicts of concurrent writes, we 
compute $P_r$'s by the following order and write them to 
corresponding quadrants of $C$ once they are done.
\begin{align*}
    P_1 \serial (P_2 \parallel P_5) \serial (P_3 \parallel P_4)
    \serial (P_6 \parallel P_7)
\end{align*}

The recurrences for time and space are then as follows.
\begin{align*}
T_\infty (n) &= 4 T_\infty (n/2) + 3 T_{\infty, \oplus} (n/2) & S (n) &= 4 (n/2)^2 + 2 S (n/2) \\
T_{\infty, \oplus} (n) &= T_{\infty, \oplus} (n/2) & & 
\end{align*}
The time recurrence says that the computing at each recursion 
level of dimension $n$ is divided into four parallel steps,
each of which can take at most $4 (n/2)^2 = n^2$ temporary space.
For an instance, the computing of $P_6$ and $P_7$ requires four
temporary matrices to hold the inputs calculated from $\oplus$ 
and $\ominus$ of quadrants of $A$ and $B$. 
The $\oplus$ operations of the results from the four parallel 
steps are also performed in a $2$-way divide-and-conquer fashion.
The parallel time and space then solve to $T_\infty (n) = O(n^2)$ 
and $S (n) = O(n^2)$.

As in the general MM (refer to \secref{tar-mm}),
the TAR technique can remove the artificial dependency introduced
by the ``$\serial$'' constructs (global synchronization) between 
parallel steps and parallelize all multiplications. However, the 
additions (as well as subtractions) to reduce the final results
can not be eliminated. So the above time recurrence counts the 
number of additions (subtractions) on the critical path, which 
again solves to $O(n^2)$. 

\begin{lemma}
    There is a TAR algorithm for Strassen MM that 
    has an $O(n^2)$ time, $O(n^{\log_2 7})$ work, 
    $O(n^2)$ space, and an $O(n^{\log_2 7}/B)$ serial 
    cache bound.
\label{lem:tar-strassen}
\end{lemma}
} % end punt

\subsecput{sar-strassen}{The SAR algorithm for Strassen}
\begin{lemma}
    There is a \proc{SAR-Strassen} algorithm that has optimal
    $O(\log n)$ time,
    $O(n^{\log_2 7})$ work, $O(p n^2)$ space, and optimal 
    $O(n^2/B + n^{\log_2 7}/(B M^{1/2 \log_2 7 - 1})$ 
    cache bound.
\label{lem:sar-strassen}
\end{lemma}
\begin{proof}
Observing that on any processor three copies of temporary 
matrices of size $n/2^k$-by-$n/2^k$ is sufficient 
for all subtasks at the depth $k$, we
have the following improved space recurrences with no change made
to the time recurrences (of \lemref{par-strassen}), 
which solve to $S (n) = p n^2$.
\begin{align*}
S (n) &= p S_1 (n) & S_1 (n) &= S_1 (n/2) + 3 (n/2)^2
\end{align*}
The three copies of temporary matrices are used to hold the
input and output matrices of $P_r = S_r \otimes T_r$ at each
depth, with $S_r$ and $T_r$ computed on the fly
from $A$ and $B$ respectively. The final quadrants of $C$ can 
be computed by reusing the space of $C$ and $P$'s.
Analogous to the \proc{SAR-MM} algorithm, the only change to
cache recurrences of \lemref{par-strassen}
is the stop condition. 
\begin{align*}
Q_1 (n) &= 7 Q_1 (n/2) + O(n^2/B) \\ 
Q_1 (n) &= O(n^2/B) \quad \text{if $(1 + 1/4 + \ldots ) 3 n^2 + 3 n^2 \leq \epsilon M$}
\end{align*}
In the stop condition, the term $(1 + 1/4 + \ldots ) 3n^2$ 
accounts for the summation of temporary space reused for the 
same-sized memory
requests on the same processor and $3 n^2$ accounts for input 
and output quadrants of $A$, $B$, and $C$.
The cache recurrences solve to the optimal 
$O(n^2/B + n^{\log_2 7}/(B M^{1/2 \log_2 7 - 1})$.
Again, this is because
the SAR algorithm enforces the reuse of temporary storage
of every depth on each processor.
\end{proof}

\subsecput{star-strassen}{The STAR algorithm for Strassen}
\begin{theorem}
    There is a \proc{STAR-Strassen} algorithm that has an 
    $O(p^{1/2} \log n)$ time, $O(p^{0.09} n^{\log_2 7})$ work, 
    optimal $O(n^2)$ space, and $O(p^{0.09} \cdot n^{\log_2 7}/(B M^{1/2 \log_2 7 - 1}) + p^{1/2} \cdot n^2/B)$
    serial cache bound.
\label{thm:star-strassen-1}
\end{theorem}
\begin{proof}
We construct the \proc{STAR-Strassen} by employing the 
\proc{TAR-MM} algorithm (\thmref{tar-mm}) at upper levels of 
recursion and switching to the \proc{SAR-Strassen} algorithm 
(\lemref{sar-strassen})
after depth \id{k}, where \id{k} is a parameter to be determined
later. 

The recurrences of the total work, time and space bounds of the hybrid algorithm become:
\begin{align*}
    T_\infty (v) &= 2 T_\infty (v/2) + O(1) & \text{if $v > n/2^{k}$}\\
    T_1 (v) &= 8 T_1 (v) & \\
    S (v) &= S (v/2) & \\ 
    Q_1 (v) &= 8 Q_1 (v/2) \\
    T_\infty (v) &= O(\log_2 v) & \text{if $v \leq n/2^{k}$}\\
    T_1 (v) &= O(v^{\log_2 7}) & \\ 
    S (v) &= O(p v^2) & \\
    Q_1 (v) &= O(v^{\log_2 7}/(BM^{1/2 \log_2 7 - 1}) + v^2/B) &
\end{align*}
Before the depth reaches \id{k}, it inherits the
same recurrences of the \proc{TAR-MM} algorithm and switches
to those of \proc{SAR-Strassen} after depth \id{k}. 
Making $\id{k} = (1/2) \log_2 p$, we have the total work 
bound of $O(p^{1/2 (3 - \log_2 7)} n^{\log_2 7} \approx 
O(p^{0.09} n^{\log_2 7})$, a factor of $O(p^{0.09})$ larger
than the original Strassen, the time bound of
$T_\infty (n) = O(p^{1/2} \log_2 (n/\sqrt{p})) = 
O(p^{1/2} \log_2 n)$, a factor of $O(p^{1/2})$ longer, 
the space bound of $S (n) = O(n^2)$
, a factor of $O(n^{\log_2 7 - 2}) \approx O(n^{0.8})$ 
improvement, and the serial cache bound of 
$O(p^{3/2 - 1/2 \log_2 7} \cdot n^{\log_2 7}/(B M^{1/2 \log_2 7 - 1}) + p^{1/2} \cdot n^2/B)$, where $3/2 - 1/2 \log_2 7 
\approx 0.09$. 
Notice that the constant hidden in the big-Oh of
space bound is just $1$. That is, besides the $3 n^2$ space 
for the input and output matrices \id{A}, \id{B}, and \id{C}, 
the \proc{STAR-Strassen} algorithm just requires a total of 
$1 n^2$ extra temporary storage.
\end{proof}

Though a factor of $O(p^{0.09})$ increase in work and cache bound
seems small in theory, in practice the increase can matter.
Since the expected running 
time of a processor-oblivious algorithm on a $p$-processor 
system under an RWS scheduler is $T_1/p + O(T_\infty)$ 
\cite{BlumofeLe99}, the $T_1/p$ term will 
dominate given sufficient parallelism as in the case of our
\proc{STAR-Strassen} whose $T_\infty = O(p^{1/2} \log_2 n)$. 

\begin{theorem}
    There is an alternative \proc{STAR-Strassen} algorithm that 
    has an optimal $O(n^{\log_2 7})$ work, optimal $O(\log n)$
    time, near-optimal 
    $O(n^{\log_2 7}/(B M^{1/2 \log_2 7 - 1}) + 
    p^{1/2 \log_2 7 - 1} n^2/B)$ 
    cache and an $O(p^{1/2 \log_2 7} n^2)$ space bound.
\label{thm:star-strassen-2}
\end{theorem}
\begin{proof}
An alternative way to construct a \proc{STAR-Strassen} algorithm 
is to employ the straightforward parallel Strassen algorithm 
(\lemref{par-strassen}) for the top \id{k} 
levels of recursion before switching to the 
\proc{SAR-Strassen} algorithm (\lemref{sar-strassen}). 
Obviously the work and time bound will stay asymptotically
optimal because both upper and lower levels are pure Strassen
with no control or data dependency inserted on critical path.
The recurrences for space and cache are as follows.
\begin{align*}
    S (v) &= 7 S (v/2) + O(v^2) & \text{if $v \geq n/2^k$}\\
    Q_1 (v) &= 7 Q_1 (v/2) + O(v^2/B) & \\
    S (v) &= O(p v^2) & \text{if $v < n/2^k$}\\
    Q_1 (v) &= O(v^{\log_2 7}/(B M^{1/2 \log_2 7 - 1}) + v^2/B) &
\end{align*}
Solving the recurrences, we have $S(n) = O(p^{1/2 \log_2 7} n^2)$
and $Q_1 (n) = O(n^{\log_2 7}/(B M^{1/2 \log_2 7 - 1}) + 
p^{1/2 \log_2 7 - 1} n^2/B)$.
\end{proof}

The alternative \proc{STAR-Strassen} in \thmref{star-strassen-2}
has asymptotically optimal work, time, and near-optimal cache
bound, so it may perform better in practice.

\punt{ % begin punt
\para{STAR technique for fast MM algorithms: }
The STAR technique should work for any Strassen-like fast
MM algorithm that trades more additions for less multiplications
such as the Strassen-Winograd \cite{DalbertoBoNi11} and so on.
All such algorithms can be 
formulated in a special model of bilinear circuit as follows
\begin{align}
    C_{i} &= \sum_{r=1}^{R(i)} \gamma_{i, r} \cdot \left (\sum_{p=1}^{P(i)} \alpha_{i, p} A_{p}\right ) \cdot \left (\sum_{q=1}^{Q(i)} \beta_{i, q} B_{q} \right ) 
\label{eq:bilinear}
\end{align}
where $C_i$, $A_p$ and $B_q$ stand for some sub-matrix of
the corresponding input / output matrices; 
$\gamma_{i, r}$, $\alpha_{i, p}$, 
and $\beta_{i, q}$ are some constants specific to the fast 
algorithm. 
Suppose that each matrix can be decomposed into 
a grid of $(n/b)^2$ equally sized sub-matrices at each recursion 
level 
(i.e. each sub-matrix will be of size $b$-by-$b$) and 
$R = \sum_{i=1}^{b^2} R(i)$.
By a straightforward parallel implementation, these 
fast algorithms have work and space bound of the form 
$S(n) = R S(b) + a (b)^2$, where $a \leq 3R$. The bound 
solves to $O(n^{\log_{n/b} R})$. All our analysis for the 
\proc{SAR-Strassen} and \proc{STAR-Strassen} will then work 
through. 
} % end punt

\secput{expr}{Experiments}

We have implemented the TAR, SAR, and STAR algorithms for
general matrix multiplication of data type ``double'', and
compared their performance with classic \proc{CO2}
and \proc{CO3} algorithms on a $24$-core shared-memory 
machine (\figref{machine-spec}). 
% The machine specification is listed in \figref{machine-spec}.

\punt{% begin punt
\begin{figure}[!ht]
    \begin{subfigure}[b]{0.9 \linewidth}
    \centering
    \includegraphics[width = \textwidth]{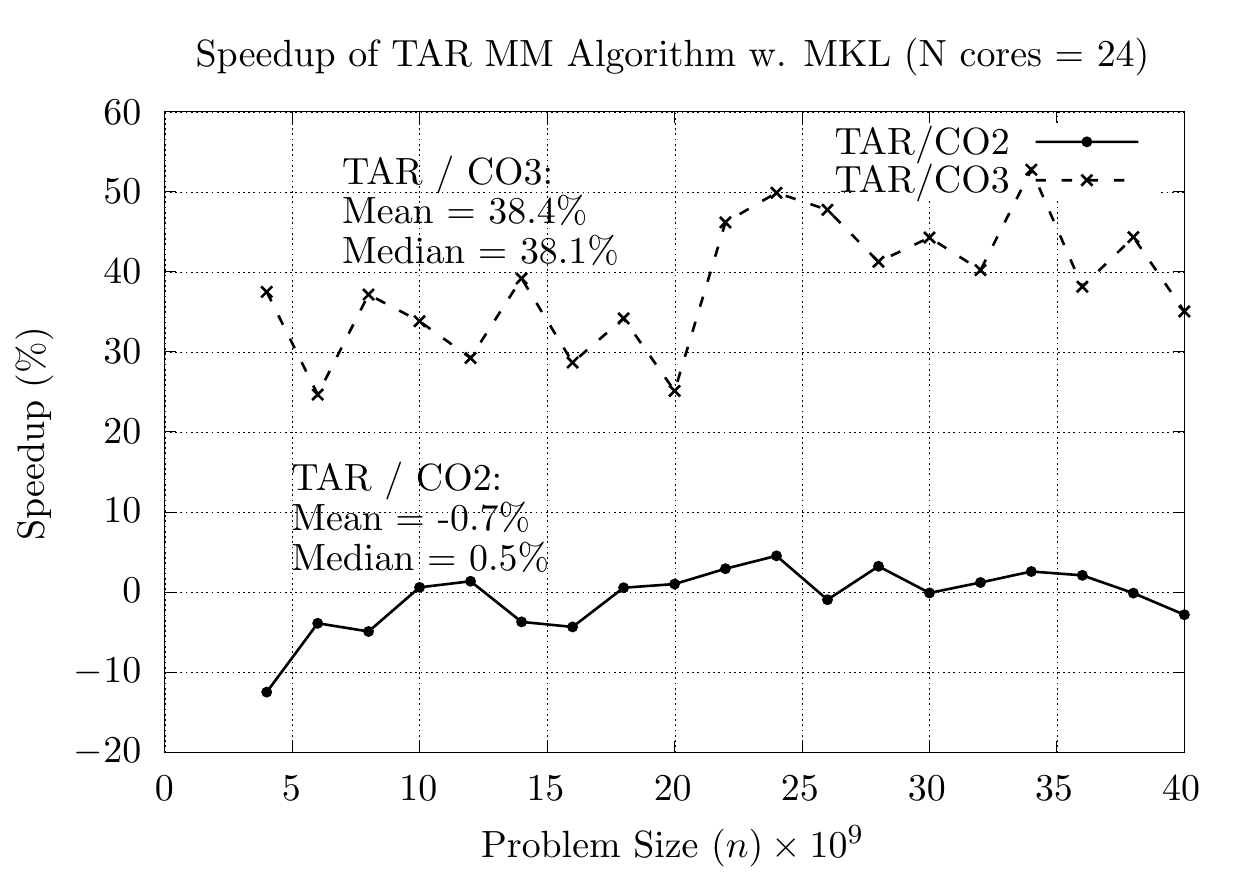}
    % \vspace*{-1\baselineskip}
    \caption{\proc{TAR-MM}'s speedup over \proc{CO2} and \proc{CO3} with MKL}
    \label{fig:24c-mm-distri}
    \end{subfigure}
    \vfil
    \begin{subfigure}[b]{0.9 \linewidth}
    \centering
    \includegraphics[width = \textwidth]{./mm_tar_raw_24c_distri.pdf}
    % \vspace*{-1\baselineskip}
    \caption{\proc{TAR-MM}'s speedup over \proc{CO2} and \proc{CO3} without MKL}
    \label{fig:24c-raw-mm-distri}
    \end{subfigure}
    % \vspace*{-1\baselineskip}
    \caption{Experiments on general MM algorithms on $24$-core machine.}
    \label{fig:24c-mm}
\end{figure}
}% end punt

\begin{figure}[!ht]
    \centering
    \includegraphics[width = .9 \linewidth]{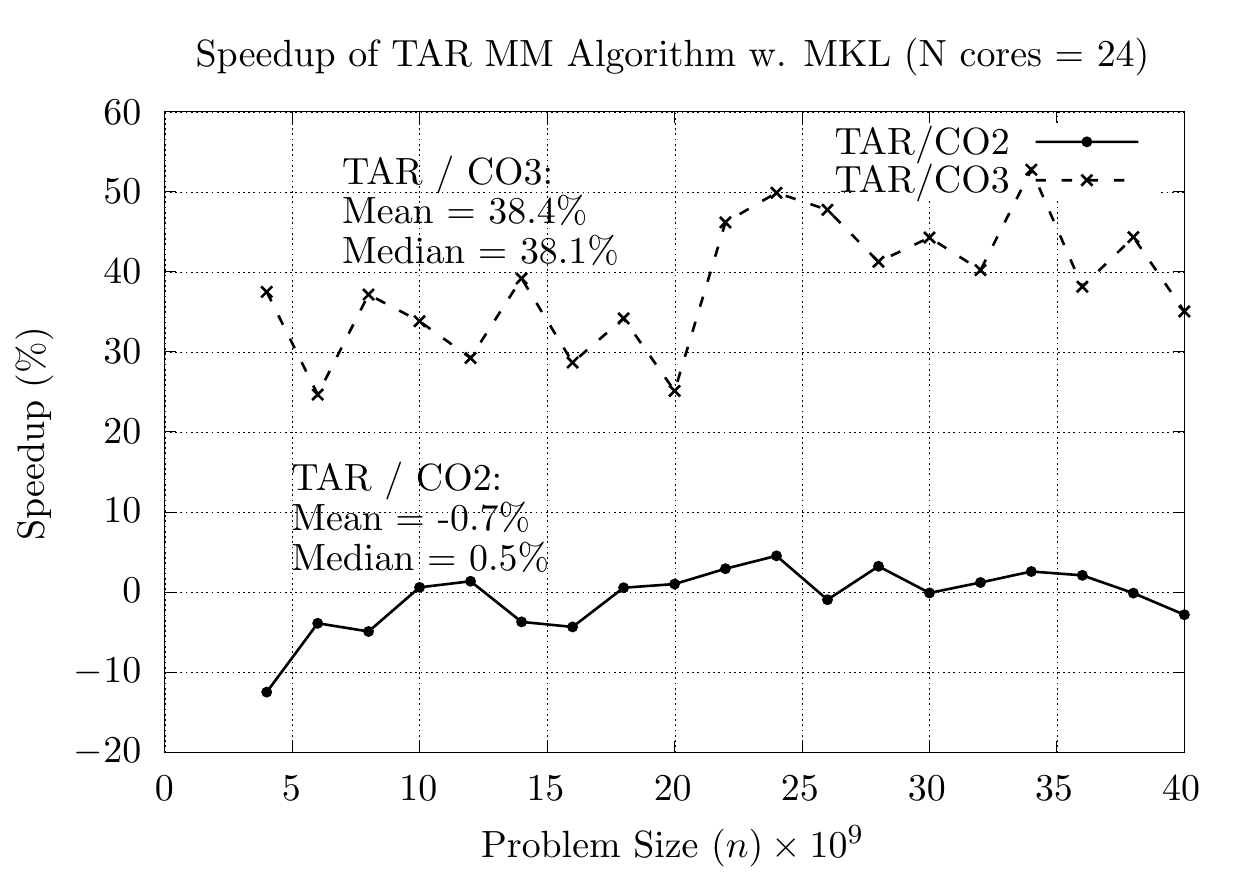}
    % \vspace*{-1\baselineskip}
    \caption{\proc{TAR-MM}'s speedup over \proc{CO2} and \proc{CO3} with MKL kernel}
    \label{fig:24c-mm-distri}
\end{figure}

\begin{figure}[!ht]
\centering
\begin{tabular}{cccc}
\toprule
\multicolumn{4}{c}{with MKL kernel}\\
\midrule
Mean/Median Spdp (\%) & TAR & SAR & STAR \\
\proc{CO2} & $-0.7$/$0.5$ & $-2.0$/$-1.0$ & $-2.6$/$-1.8$ \\
\proc{CO3} & $38.4$/$38.1$ & $36.6$/$36.5$ & $35.8$/$34.0$ \\
\midrule
\multicolumn{4}{c}{with manual kernel}\\
\midrule
Mean/Median Spdp ($\%$) & TAR & SAR & STAR \\
\proc{CO2} & $11.8$/$9.2$ & $9.3$/$6.8$ & $10.5$/$8.0$ \\
\proc{CO3} & $1.6$/$1.5$ & $-0.6$/$-0.5$ & $0.5$/$0.5$ \\
\bottomrule
\end{tabular}
\caption{Mean and Median Speedup of TAR, SAR, and STAR algorithms
over \proc{CO2} and \proc{CO3} with MKL kernel and with manually
implemented kernel.
All numbers shown in cell are in percentage. For an instance,
In above rows of ``with MKL kernel'', the cell in intersection
of TAR and \proc{CO2} reads ``$-0.7$/$0.5$'', which means with MKL
kernel, the mean speedup of TAR algorithm over \proc{CO2} is 
$0.7\%$ slower, while the median is $0.5\%$ faster.}
\label{fig:24c-mm-perf-table}
\end{figure}

\begin{figure}[!ht]
\caption{Experiementing Machine
%\footnote{The sizes of L1 icache on both machines are identical
%to that of dcache.}
}
\label{fig:machine-spec}
\begin{tabular}{c|c}
\toprule
Name & $24$-core machine \\
\midrule
OS   & CentOS 7 x86\_64 \\
Compiler & ICC 19.0.3 \\
% GCC Compatibility & 4.8.5  \\
CPU type & Intel Xeon E5-2670 v3 \\
Clock Freq & 2.30 GHz \\
\# sockets & 2 \\
\# cores / socket & 12 \\
Dual Precision & \\
FLOPs / cycle & 16 \\
Hyper-Threading & disabled \\
L1 dcache / core & 32 KB \\
% L1 icache / core & 32 KB \\
L2 cache / core & 256 KB \\
L3 cache (shared) & 30 MB \\
memory & 132 GB \\
\bottomrule
\end{tabular}
\end{figure}

In order for a \emph{Fair} performance comparison,
we require that all competing algorithms call the same 
kernal function for serial base-case computation
in the same round of comparison.
\punt{% begin punt
That is to say, all competing algorithms call Intel MKL's 
serial \func{dgemm} \footnote{\func{dgemm} computes 
double-precision matrix multiplication of 
$C = \alpha X \times Y + \beta C$,
where $A$, $B$, and $C$ are matrices and $\alpha$ and $\beta$
are scalars.} 
and \func{daxpy} \footnote{\func{daxpy} computes double-precision
vector addition of $Y = \alpha \times X + Y$, 
where $X$ and $Y$ are vectors and $\alpha$ is a scalar.}
subroutines for 
base-case matrix multiplication and addition in one round 
of performance comparison and a straightforward implemented 
matrix multiplication and addition in another round.
We have to mention that Intel MKL does not have any
subroutine for matrix addition, while \func{daxpy} is for 
vector addition. If we call \func{daxpy} multiple times for 
matrix addition, it will increase function-call
overheads. So the intermediate results of our new
algorithms are stored in Z-morton order, we restore it to 
row-major order by the end of algorithm with corresponding 
time included in overall timing.
}% end punt
Moreover, we mandate the same base-case size for serial 
computation in all competing algorithms.
Thus, the only difference among competing
algorithms is how they partition and schedule tasks.
\punt{% begin punt
Another point we have to mention is that all our new algorithms
, as well as \proc{CO3}, have an additional $O(n^3)$ overheads
of \func{daxpy} over \proc{CO2} because \proc{CO2} only calls
\func{dgemm}.

Due to space limit, we show the full comparing results 
of \proc{TAR-MM} algorithm over \proc{CO2} and \proc{CO3} 
in \figref{24c-mm}
, as well as mean and median speedups of other algorithms 
in \figref{24c-mm-perf-table}.
}% end punt
We compute speedup by 
``$(\text{running\_time}_{\text{peer alg.}} / 
\text{running\_time}_{\text{STAR}} - 1) \times 100\%$''.

\paragrf{Brief Summary: }
Our TAR algorithm performs consistently the fastest even with
an $O(n^3)$ additional overheads of \func{daxpy}; Other new
algorithms are generally faster than \proc{CO2} and \proc{CO3},
especially when problem
dimension is reasonably large. Interestly, with a relative
faster kernel, i.e. MKL's \func{dgemm} and \func{daxpy}, 
\proc{CO2} is faster than \proc{CO3}. While with a slower 
manually implemented kernel, \proc{CO3} becomes faster. 
More experimenting data can be presented in a full version.

\punt{% begin punt
\paragrf{More Details: }
\begin{enumerate}
    \item \afigref{24c-mm-distri} shows that if all competing 
        algorithms call a relatively faster kernel, i.e. Intel 
        MKL's \func{dgemm} and \func{daxpy} for base-case 
        computation, the mean speedup of our TAR algorithm 
        over \proc{CO2} is $-0.7\%$, i.e. $0.7\%$ slower in 
        average, while the median is $0.5\%$ faster;
        while the mean speedup over \proc{CO3} is $38.4\%$ and
        the median is $38.1\%$.

    \item Surprisingly, if all competing algorithms call a 
        relatively slower kernel, i.e. a straightforward 
        implementation of 
        matrix multiplication and addition kernel, we can see
        that \proc{CO3} algorithm becomes 
        faster than \proc{CO2}, with our TAR algorithm be the 
        fastest. We have double-checked all implementations and 
        repeated the experiments for multiple times, which all 
        confirmed this result.

        A possible explanation may be that with a slower kernel, 
        the critical path actually becomes longer, thus runtime 
        scheduler requires more subtasks, i.e. more parallelism,
        for dynamic load balance, which is the
        strength of \proc{CO3} and TAR algorithm. Moreover,
        current heap (interfaced by \func{calloc}/\func{free})
        or TLS (thread-local storage, interfaced by 
        \func{scalable\_calloc}/\func{scalable\_free}) seems to 
        have implemented the Last-In First-Out property as we
        discussed in \secref{tar-mm} by some degree so that memory 
        blocks of the same size / same recursion level are 
        largely reused in the \proc{CO3} code (\figref{mm-n3}).
        We also verified that
        \thmref{general-busy-leaves} (\secref{sar-mm}) holds
        in Cilk runtime system, i.e.
        at most $p$ subtasks of the
        same depth can exist in a $p$-processor system. 
        So the actual cache misses of \proc{CO3} is not as bad as
        $O(n^3/B)$ as generally calculated by 
        \eqreftwo{mm-n3-cache}{mm-n3-cache-stop} in \secref{intro}.

        As discussed in \secref{discussion}, our \proc{SAR-MM}
        algorithm differs from \proc{CO3} in its lazy allocation
        strategy.
        \afigref{24c-mm-perf-table} shows that this lazy 
        strategy generally pays off. Compared SAR with TAR and
        STAR, we can see that a further reduction in space and
        cache (even by a constant fraction) is necessary to get
        a better result.
\end{enumerate}
}% end punt

\secput{relWork}{Conclusion and Related Works}

\paragrf{Concluding Remarks: }
In this paper, we reviewed and analyzed classic 
matrix multiplication algorithms, \proc{CO2} and \proc{CO3},
in modern processor-oblivious runtime. The \proc{CO2} algorithm
has provably best work, space, and serial cache bounds, while
its longer critical-path length may incur more parallel cache
missees in a parallel setting. On the contrary, people used to
over-estimated \proc{CO3} algorithm's space and cache 
requirements.
We show that by the busy-leaves property \cite{BlumofeJoKu95}, 
we can derive \thmref{general-busy-leaves}, which is verified 
in popular Cilk runtime, such that there are no more than $p$
subtasks of the same depth in a $p$-processor system.
Moreover, by the Last-In First-Out memory allocation
strategy, which seems to be true in modern heap and TLS
(Thread-Local Storage) allocators, memory blocks are largely
reused in the \proc{CO3} algorithm.
By the above two properties, the classic \proc{CO3} algorithm
does not perform too bad compared with \proc{CO2} algorithm
in modern processor-oblivious runtime such as Intel Cilk Plus.
Interestingly, when employed a manually implemented (slower) 
kernel for
base-case computation, \proc{CO3} algorithm can even be faster
than \proc{CO2} probably due to its shorter critical-path length.
To further reduce space requirement or in other words maximize 
space reuse, thus minimize un-necessary cold cache misses, 
we propose a ``lazy allocation'' strategy for our
SAR and STAR algorithms.
Though we utilize atomic operations to check subtask's status
for lazy allocation, we believe that it's possible for
runtime system to provide such facility, which
can be our future work.
We also show how to possibly extend our approach to 
Strassen-like fast algorithms.

\paragrf{Related Works: }
\punt{% begin punt
Galil and Park \cite{GalilGi89, GalilPa94} proposed to solve DP
recurrences with more than $O(1)$ dependency by the methods of 
closure, matrix product, and indirection. 
% \punt{% begin punt
Maleki et al. \cite{MalekiMuMy14} presented in their paper
that certain dynamic programming problem called ``Linear-Tropical
Dynamic Programming (LTDP)'' can possibly obtain extra parallelism
by breaking data dependencies between stages.
% }% end punt
Chowdhury and Ramachandran \cite{ChowdhuryRa08} considered
a processor-aware approach that makes a different 
partitioning at different levels of recursion.
Tang et al. \cite{TangYoKa15} proposed Eager and Lazy 
cache-oblivious
wavefront (COW) technique to balance cache efficiency
and parallelism for a large class of DP algorithms.
Dinh et al. \cite{DinhSiTa16} formalized the approach to
ND model.  
}% end punt

% 3. Automatic generation of Strassen-like algorithm, which 
% doesn't help the space bound
Various optimizations, tradeoffs among work, 
space, time, communication bounds, on the general MM
on a semiring or Strassen-like fast algorithm 
has been studied for decades, including at 
least \cite{BensonBa15, BallardDeHo12, McCollTi99, SolomonikDe11,
KumarHuSa95, BoyerDuPe09, HuangSmHe16, SmithGeSm14}.
The basic idea behind these prior works is to switch 
manually back and forth between
a serial algorithm to save and reuse space and a parallel
algorithm to increase parallelism. Our approach differs
in the following ways. Firstly, our approach is dynamic 
load-balance, which is arguably more flexible and adaptive
on shared-memory system; 
Moreover, 
by generalizing the ``busy-leaves'' property of runtime 
schedulers, our technique upper-bounds space requirement to be 
asymptotically optimal without tuning; By a ``Last-In First-Out''
memory allocator and lazy allocation strategy, we bound 
cache efficiency to be asymptotically optimal without tuning; 
By having a sublinear critical-path length, we reduce
asymptotically parallel cache misses. 

Smith et al. \cite{SmithGeSm14} noticed divots in the 
performance curves of Intel MKL's \func{dgemm} when 
multiplying a matrix of size $m$-by-$k$ with a matrix of
size $k$-by-$n$ (a rank-$k$ update), where
both $m$ and $n$ are fixed to $14400$ and $k$ is very slightly
larger than a multiple of $240$.
We find the divots of Intel MKL's \func{dgemm} in 
square matrix multiplication where problem dimension is powers 
of two.

\punt{ % begin punt
They concluded the reason of
the divots as
the division of the rank-$k$ update into an optimal rank-$k_c$
update followed by a very small rank, which is very expensive.
They fixed the divots by performing only a single rank-$k$ update
with a $k$ that is larger than the optimal $k_c$. 
} % end punt

\punt{% begin punt
Shun et al. \cite{ShunBlFi13} alleviates the problem of
``concurrent writes'' to the same memory location 
by ``priority updates''. However, not all operations can be 
prioritized or reduced as in the case of general MM.
}% end punt

% \input{concl}

% \input{ack}

% We recommend abbrvnat bibliography style.

% \bibliographystyle{abbrv}
\bibliographystyle{IEEEtran}
\bibliography{papers}

% \appendix
% \input{paren}
% \input{star-lws}
%\input{star-gap}
%\input{star-gap-app}

\end{document}